\documentclass[11pt]{article}
\usepackage[letterpaper,margin=1in]{geometry}
\usepackage{xcolor}
\usepackage[colorlinks]{hyperref}
\definecolor{lightblue}{rgb}{0.5,0.5,1.0}
\definecolor{darkred}{rgb}{0.5,0,0}
\definecolor{darkgreen}{rgb}{0,0.5,0}
\definecolor{darkblue}{rgb}{0,0,0.5}

\hypersetup{colorlinks,linkcolor=darkred,filecolor=darkgreen,urlcolor=darkred,citecolor=darkblue}

\usepackage{tabularx}
\usepackage{mathtools}
\usepackage{verbatim}
\usepackage{enumitem}
\usepackage{amsthm}
\usepackage{amssymb}
\usepackage{amsmath}
\usepackage{nicefrac}
\usepackage{mathtools}
\usepackage[linesnumbered,ruled,commentsnumbered,longend]{algorithm2e}
\usepackage{tikz}
\usetikzlibrary{graphs}
\usetikzlibrary{shapes.misc, positioning}
\usetikzlibrary{shapes,fit}
\usepackage[edges]{forest}
\usetikzlibrary{shapes,arrows,calc,decorations.pathmorphing}
\usepackage{amsmath}

\newcommand{\Traces}{\textsc{Traces}}
\newcommand{\nauty}{\textsc{nauty}}
\newcommand{\bliss}{\textsc{bliss}}

\SetKwProg{Fn}{function}{}{end}\SetKwFunction{RandomAut}{Automorphisms}
\SetKwProg{Fn}{function}{}{end}\SetKwFunction{RandomIso}{Isomorphism}
\SetKwFunction{FirstLeaf}{FirstLeaf}
\SetKwFunction{RandomWalk}{RandomWalk}
\SetKwFunction{NextNodeDFS}{OneNodeDFS}
\SetKwFunction{NextNodeBFS}{OneNodeBFS}
\SetKwFunction{Some}{Some}
\SetKwFunction{Sift}{Sift}
\SetKwFunction{RandomElement}{RandomElement}
\SetKwFunction{RandomChild}{RandomChild}
\SetKwFunction{NextChild}{NewChild}
\SetKwFunction{NextRandomChild}{NewRandomChild}
\SetKwFunction{Refine}{Refine}
\SetKwFunction{Subtree}{Subtree}
\SetKwFunction{RRef}{Ref}
\SetKwFunction{SSel}{Sel}
\SetKwFunction{Sift}{Sift}
\SetKwFunction{CertifyAutomorphism}{CertifyAutomorphism}
\SetKwFunction{CertifyIsomorphism}{CertifyIsomorphism}

\newcommand{\problemdef}[2]{\noindent\textbf{Problem {#1}} ({#2})\textbf{.}}

\DeclareMathOperator{\inj}{Inj}   



\DeclareMathOperator{\Col}{col}

\DeclarePairedDelimiter{\ceil}{\lceil}{\rceil}
\newtheorem{lemma}{Lemma}
\newtheorem{corollary}[lemma]{Corollary}
\newtheorem{theorem}[lemma]{Theorem}

\newtheorem*{invarax}{Axiom}

\title{Search Problems in Trees with Symmetries:\\ \Large{near optimal traversal strategies for\\ individualization-refinement algorithms}}
\author{Markus Anders and Pascal Schweitzer\\ TU Kaiserslautern}

\newcommand\blfootnote[1]{%
  \begingroup
  \renewcommand\thefootnote{}\footnote{#1}%
  \addtocounter{footnote}{-1}%
  \endgroup
}
\begin{document}

\maketitle

\begin{abstract}
We define a search problem on trees that closely captures the backtracking behavior of all current practical graph isomorphism algorithms. Given two trees with colored leaves, the goal is to find two leaves of matching color, one in each of the trees. The trees are subject to an invariance property which promises that for every pair of leaves of equal color there must be a symmetry (or an isomorphism) that maps one leaf to the other.

We describe a randomized algorithm with errors for which the number of visited leaves is quasilinear in the square root of the size of the smaller of the two trees. For inputs of bounded degree, we develop a Las Vegas algorithm with a similar running time.

We prove that these results are optimal up to logarithmic factors. We show a lower bound for randomized algorithms on inputs of bounded degree that is  the square root of the tree sizes. For inputs of unbounded degree, we show a linear lower bound for Las Vegas algorithms. For deterministic algorithms we can prove a linear bound even for inputs of bounded degree. This shows why randomized algorithms outperform deterministic ones.

Our results explain why the randomized ``breadth-first with intermixed experimental path'' search strategy of the isomorphism tool \Traces{} (Piperno 2008) is often superior to the depth-first search strategy of other tools such as \nauty{} (McKay 1977) or \bliss{} (Junttila, Kaski 2007). However, our algorithm also provides a new traversal strategy, which is theoretically near optimal with better worst case behavior than traversal strategies that have previously been used.
\end{abstract}\blfootnote{The research leading to these results has received funding from the European Research Council (ERC) under the European Union's Horizon 2020 research and innovation programme (EngageS: grant agreement No.~{820148}).}

\thispagestyle{empty}

\newpage
\setcounter{page}{1}

\section{Introduction}
We define a new search problem involving trees with symmetries. In this problem, two unknown trees are given as input and they can be gradually explored. The leaves of the trees are colored and the task is to find a pair of leaves, one in each tree, with matching colors or determine that such a pair does not exist. The crucial element that distinguishes our model from standard exploration tasks is that the color of the leaves allows us to draw conclusions about the local surroundings of the leaf. More precisely, there is an invariance axiom guaranteed to always hold. It says that if two leaves are of the same color then there is a symmetry (an automorphism or an isomorphism depending on the leaves being in the same tree or not) that maps one leaf to the other. 

The invariance property guarantees that the local neighborhood around the leaf is structurally the same as the neighborhood around other leaves of the same color. This allows us to discard unexplored parts of the search tree thereby opening the possibility of having algorithms that explore only a sublinear number of the nodes of the trees.

Our motivation behind the model lies in the desire for a theoretical analysis of practical solvers for the graph isomorphism problem,  automorphism group computations and graph canonization. While the best theoretical algorithms, such as Babai's quasipolynomial time algorithm~\cite{DBLP:conf/stoc/Babai16}, are based on algorithmic group theory, practical solvers \cite{saucy:webpage,JunttilaKaski:ALENEX2007, conauto:webpage, McKay:userguide,  DBLP:journals/corr/abs-0804-4881} exclusively follow the individualization-refinement (IR) paradigm.
First introduced into the realm of practical isomorphism testing and canonization by McKay in 1977~\cite{McKay81practicalgraph}, the basic principle has remained unchanged to date.
Indeed, all algorithms in this paradigm perform some form of backtracking, implicitly creating a recursion tree for each input graph. 
The tools solve the graph isomorphism problem by finding two leaves in this recursion tree that correspond to each other. The number of nodes in the recursion tree that are actually called during the execution is closely linked to the running time of the overall algorithm.
Despite being simple, our search problem and the exploration model capture precisely the task needed to be solved and assess correctly the running times of solutions.

Traditionally, IR-tools have followed a depth-first search approach to traverse the search tree. However, in 2008, Piperno~\cite{DBLP:journals/corr/abs-0804-4881} introduced his tool \Traces{}, which broke away from this principle, performing a form of breadth-first search that is intermixed with random traversal of the tree (so-called experimental paths). The traversal strategy is at the very heart of the underlying algorithm.
Significantly outperforming all the other tools on most practical inputs~\cite{McKay201494}, \Traces{} revealed that the traversal strategy is arguably the most important design choice in IR algorithms. This immediately raises the question whether there are theoretical, structural reasons why this traversal strategy is favorable. Going one step further, we can ask for optimal traversal strategies. 

However, so far there has been no rigorous justification as to why one traversal strategy should be superior to another and in particular there is no research into optimal traversal strategies.

\paragraph{Contribution.}
The introduction of our particular search problem in trees with symmetry allows us to strip away all the other design choices that have to be made in the creation of an IR-algorithm and isolates the core issue of the traversal strategy in the search tree.
The ultimate goal is twofold. Firstly, to provide a more rigorous, theoretical foundation for the design of practical graph isomorphism tools. Secondly, to design novel, near optimal strategies to be adopted in future generations of practical solvers. 

An input consists of two trees without vertices that only have one child. 
Let $n$ denote the size of the smaller one of the two trees and $N$ the size of the larger one. 
The cost of our algorithms is measured in the number of nodes that are explored. For our algorithms, the terms ``running time'' or ``complexity'' refer to this cost measure.

Regarding \emph{upper bounds}, we provide a simple randomized Monte Carlo algorithm with an upper bound of~$\mathcal{O}(\sqrt{n} \log n)$ explored nodes. 
For trees of bounded degree we design a more complicated algorithm achieving an upper bound of~$\mathcal{O}( \sqrt{n} \log N)$ for Las Vegas algorithms (i.e., randomized algorithms without errors). 

These algorithms are accompanied by nearly matching \emph{lower bounds}, showing that~$\Omega(\sqrt{n})$ nodes need to be explored for randomized Monte Carlo algorithms even on bounded degree trees and that for unbounded degree inputs, Las Vegas algorithms need to visit $\Omega(n)$ nodes in expectation. 
For deterministic algorithms we get a lower bound of $\Omega(n)$ even for inputs of bounded degree.

Figure~\ref{fig:bounds} provides an overview of the lower and upper bounds proven in this paper.

Overall this shows that the new traversal strategies are optimal up to logarithmic factors. 
However, it also shows that randomized traversal strategies, even those without error, asymptotically outperform deterministic ones.

The algorithms for proving the upper bounds immediately imply individualization-refinement algorithms with the same runtimes (up to almost linear factors in the order of the graphs for non-recursive work). We should emphasize that our new upper bounds asymptotically outperform the traversal strategies that are currently being used in practice in the worst case.

\begin{figure}
\centering
\begin{tabularx}{0.73\linewidth}{|l|l|l|X|}
	\hline
	\footnotesize \textbf{Setting}     & \footnotesize \textbf{Lower Bound}   & \footnotesize \textbf{Lower Bound ($d$-adic)} & \footnotesize \textbf{Upper Bound} \\\hline
	\footnotesize Monte Carlo & \footnotesize $\Omega(\sqrt{n})$ & \footnotesize $\Omega(\sqrt{n})$ 
	& \footnotesize $\mathcal{O}(\log(n)\sqrt{n})$ \\\hline
	\footnotesize Las Vegas     & \footnotesize $\Omega(n)$ & \footnotesize $\Omega(\sqrt{n})$ 
	& \footnotesize $\mathcal{O}(d \log(N) \sqrt{n})$\\\hline
	\footnotesize Deterministic & \footnotesize $\Omega(n)$ & \footnotesize $\Omega(n)$ & \footnotesize $\mathcal{O}(n)$ \\\hline
\end{tabularx}
\caption{This table summarizes lower and upper bounds for the isomorphism problem implied by the results of this paper. Here~$n= \min\{n_1,n_2\}$ and~$N= \max\{n_1,n_2\}$, where the sizes of the trees are $n_1, n_2$ and $d$ gives the maximum degree of the two input trees.
We state separate lower bounds for trees with bounded ($d$-adic) and unbounded degree.} \label{fig:bounds}
\end{figure}

\section{A Model for Tree Exploration with Symmetries}

This section presents the exploration model that is used throughout this paper. 
The model enables us to perform a focused analysis of the traversal strategies used in the search trees of individualization-refinement algorithms.
We state the model independent of any discussion regarding individualization-refinement.
In Section~\ref{sec:motivation} we explain why the model captures the running time of individualization-refinement algorithms.

\subsection{Black Box Search Trees}
In our search problems the input consists of one or two hidden trees of which certain information is to be discovered. We first explain how the trees can be explored.

\paragraph{Exploration Model.} 

We consider rooted trees in which there is a priori no bound on the degree of the vertices. However, we require that no vertex has exactly one child. Furthermore, the leaves of the tree are colored.

Our exploration model for the trees restricts access of algorithms to the trees themselves and how they can be explored. We think of new information as being provided by an oracle to the exploration algorithm.

During execution, a node of the tree is either \emph{explored} or \emph{unexplored}.
Whenever a new node is explored, the algorithm learns the number of children of the node. In particular the algorithm knows whether the node is a leaf or not. Furthermore, in case the node is a leaf, it learns its color.

At the beginning of an execution, everything except the root is deemed unexplored. The algorithm can only ever access previously explored nodes.
The degree (i.e., the number of children) of an explored node $v$, which we denote by $\deg(v)$, is always known.
To explore further nodes, the algorithm can explore a child of a previously explored node.
Specifically, the algorithm can request the~$i$-th child of~$v$ with~$i\in \{1,\ldots,d(v)\}$, which thereby becomes explored. For this the input has an arbitrary but fixed ordering for the children of each vertex.

The \emph{cost} of the exploration is measured in the number of oracle accesses, i.e., the number of nodes that are ever visited by the algorithm. (In particular there is no cost for traversing previously explored parts of the tree.)
 
Figure~\ref{fig:blackboxex} illustrates such an exploration of a tree. 
Note that while the algorithm always knows the degree of explored nodes, it is essentially unable to chose a new specific child to explore since in another input the ordering of the children may be different.

More formally, a black box search tree $T=(V, E, \Col)$ consists of a rooted tree with colored leaves and for each node an ordering of the children. We omit the orderings from the notation. Of course all choices of orderings lead to proper search trees.
The function $\Col\colon L(T) \to \mathbb{N}$ maps the \emph{leaves of the tree} $L(T)$ to natural numbers, which we will refer to as colors.

In our algorithms, we use the procedure $\NextChild \colon V \to V \cup \{\bot\}$ to explore the tree, which agrees with the previous description as follows. For an explored vertex $v$ the algorithm chooses the smallest index of a previously unexplored child of $v$ and queries the oracle for that child. If no unexplored child exists, the function returns $\bot$.  

In the description of randomized algorithms, we also use the function $\RandomChild\colon  V \to V\cup \{\bot\}$, which returns a child chosen uniformly at random among all children of $v$, which means that it can in particular return previously explored children. 

\newcommand\eT{T}

\newcommand\ranelem{\stackrel{\mathclap{\normalfont\tiny\mbox{R}}}{\in}}

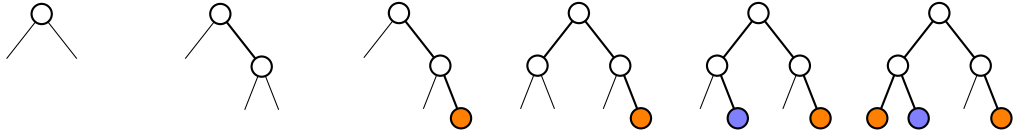
\begin{figure}
	\centering
	\begin{tikzpicture}[
		every node/.style = {minimum width = 0.7em, inner sep = 0, outer sep = 0, draw, circle, fill=gray!50},
		level/.style = {sibling distance = 11mm/#1, level distance=7mm}
		]
		\node[thick, draw=black, fill=white] {}
		child {	node[draw=none, fill=none] {}
				child {	node[draw=none, fill=none] {}
				edge from parent[draw=none]}
				child {	node[draw=none, fill=none] {}
				edge from parent[draw=none]}
		}
		child {	node[draw=none, fill=none] {}
				child {	node[draw=none, fill=none] {}
				edge from parent[draw=none]}
				child {	node[draw=none, fill=none] {}
				edge from parent[draw=none]}
		};
	\end{tikzpicture} \hspace{0.2cm}
	\begin{tikzpicture}[
		every node/.style = {minimum width = 0.7em, inner sep = 0, outer sep = 0, draw, circle, fill=gray!50},
		level/.style = {sibling distance = 11mm/#1, level distance=7mm}
		]
		\node[thick, fill=white] {}
		child {	node[draw=none, fill=none] {}
				child {	node[draw=none, fill=none] {}
				edge from parent[draw=none]}
				child {	node[draw=none, fill=none] {}
				edge from parent[draw=none]}
		}
		child {	node[thick, fill=white] {}
				child {	node[fill=none, draw=none, line width=0.4pt] {}
						edge from parent[line width=0.4pt]
				}
				child {	node[fill=none, draw=none, line width=0.4pt] {}
						edge from parent[line width=0.4pt]
				}
				edge from parent[thick]
		};
	\end{tikzpicture} \hspace{0.2cm}
	\begin{tikzpicture}[
		every node/.style = {minimum width = 0.7em, inner sep = 0, outer sep = 0, draw, circle, fill=gray!50},
		level/.style = {sibling distance = 11mm/#1, level distance=7mm}
		]
		\node[thick, fill=white] {}
		child {	node[draw=none, fill=none] {}
				child {	node[draw=none, fill=none] {}
				edge from parent[draw=none]}
				child {	node[draw=none, fill=none] {}
				edge from parent[draw=none]}
		}
		child {	node[thick, fill=white] {}
				child {	node[fill=none, draw=none, line width=0.4pt] {}
						edge from parent[line width=0.4pt]
				}
				child {	node[fill=orange,thick] {}
						edge from parent[thick]
				}
				edge from parent[thick]
		};
	\end{tikzpicture} \hspace{0.2cm}
	\begin{tikzpicture}[
		every node/.style = {minimum width = 0.7em, inner sep = 0, outer sep = 0, draw, circle, fill=gray!50},
		level/.style = {sibling distance = 11mm/#1, level distance=7mm}
		]
		\node[thick, fill=white] {}
		child {	node[thick, fill=white] {}
				child {	node[draw=none, fill=none] {}
				edge from parent[line width=0.4pt]}
				child {	node[draw=none, fill=none] {}
				edge from parent[line width=0.4pt]}
				edge from parent[thick]
		}
		child {	node[thick, fill=white] {}
				child {	node[fill=none, draw=none, line width=0.4pt] {}
						edge from parent[line width=0.4pt]
				}
				child {	node[fill=orange, thick] {}
						edge from parent[thick]
				}
				edge from parent[thick]
		};
	\end{tikzpicture} \hspace{0.2cm}
	\begin{tikzpicture}[
		every node/.style = {minimum width = 0.7em, inner sep = 0, outer sep = 0, draw, circle, fill=gray!50},
		level/.style = {sibling distance = 11mm/#1, level distance=7mm}
		]
		\node[thick, fill=white] {}
		child {	node[thick, fill=white] {}
				child {	node[draw=none, fill=none] {}
				edge from parent[line width=0.4pt]}
				child {	node[fill=lightblue] {}
				edge from parent[thick]}
				edge from parent[thick]
		}
		child {	node[thick, fill=white] {}
				child {	node[fill=none, draw=none, line width=0.4pt] {}
						edge from parent[line width=0.4pt]
				}
				child {	node[fill=orange, thick] {}
						edge from parent[thick]
				}
				edge from parent[thick]
		};
	\end{tikzpicture} \hspace{0.2cm}
	\begin{tikzpicture}[
		every node/.style = {minimum width = 0.7em, inner sep = 0, outer sep = 0, draw, circle, fill=gray!50},
		level/.style = {sibling distance = 11mm/#1, level distance=7mm}
		]
		\node[thick, fill=white] {}
		child {	node[thick, fill=white] {}
				child {	node[thick, fill=orange] {}
				edge from parent[thick]}
				child {	node[fill=lightblue] {}
				edge from parent[thick]}
				edge from parent[thick]
		}
		child {	node[thick, fill=white] {}
				child {	node[fill=none, draw=none, line width=0.4pt] {}
						edge from parent[line width=0.4pt]
				}
				child {	node[fill=orange, thick] {}
						edge from parent[thick]
				}
				edge from parent[thick]
		};
	\end{tikzpicture}
	\caption{Example exploration in the black box search tree model. Starting from the root, the algorithm only ever knows explored nodes and their degrees. Through the use of an oracle, random children of explored nodes may then be queried.} \label{fig:blackboxex}
\end{figure}

\paragraph{Isomorphism Invariance.} So far we are lacking the crucial part of the model, namely symmetries. 
The core property of our trees is that the presence of leaves with equal colors implies the existence of symmetries of the trees. 
More specifically, they imply \emph{color-preserving isomorphisms}, defined as follows.
An isomorphism $\varphi$ between two trees $T_1$ and $T_2$ is a bijection on vertices $\varphi\colon V(T_1) \to V(T_2)$, such that $v$ is a child of $v'$ if and only if $\varphi(v)$ is a child of $\varphi(v')$.
A color-preserving isomorphism furthermore requires that $\Col(l) = \Col(\varphi(l))$ holds for all leaves $l \in L(V_1)$. This implies that leaves of a color can only be mapped to leaves of that same color. If $T_1 = T_2$ we also call $\varphi$ an \emph{automorphism} or a symmetry.

The crucial property that we require for all black box search trees is that whenever two leaves have the same color, we can derive an isomorphism:

\begin{invarax}[Complete Isomorphism Invariance] \label{prop:auto} If $l_1 \in T_1, l_2 \in T_2$ and $\Col(l_1) = \Col(l_2)$, then there exists a color-preserving isomorphism $\varphi\colon V(T_1) \to V(T_2)$ such that $\varphi(l_1) = l_2$.  
\end{invarax}

We should highlight that the axiom in particular has to hold for the case $T_1 = T_2$, yielding automorphisms (possibly the identity if~$l_1=l_2$).

The crucial consequence of the axiom is that it allows us to draw conclusions about the structure of unexplored parts of the search tree. For example, applying this knowledge enables us to conclude that the last remaining node of Figure~\ref{fig:blackboxex} is blue.

In the rest of the paper, we assume that all inputs, may they consist of one or two trees, adhere to this axiom. Also, all exploration algorithms operate in the exploration model we defined.

\paragraph{Isomorphism  Exploration Problem.} We are now ready to the state our problem of interest for black box search trees: the isomorphism exploration problem. 

\medskip

\problemdef{}{Isomorphism exploration} Given two search trees $T_1$ and $T_2$, compute leaves $l_1 \in T_1$ and $l_2 \in T_2$ with $\Col(l_1) = \Col(l_2)$, if they exist and return $\bot$ otherwise.
\medskip

For simplicity we will always assume that the trees are disjoint, that is $V(T_1) \cap V(T_2) = \varnothing$. This way we do not need to specify for oracle queries what tree they relate to. There are other interesting problems, such as finding two leaves of the same color within one tree, that can be defined within the model. We discuss them and their relationship to the problem just defined in Section~\ref{sec:motivation}. 

While the tree model and the exploration problem can be defined it their own right, we have a concrete motivation behind the definitions. Specifically, the motivation behind the specifics of our model lies in so-called \emph{individualization-refinement} algorithms, the prevailing method to solve the graph isomorphism problem in practice. In fact, the isomorphism exploration problem captures very closely the runtime of these algorithms. A more detailed explanation of this is differed to the end of the paper in Section~\ref{sec:motivation}.

However, let us briefly remark that in our context the requirement that inner nodes may not have exactly one child is not only natural for our intended application but also crucial. If we drop this requirement the nature of the problem changes dramatically and ray searching as well as doubling techniques become relevant (see~\cite{DBLP:journals/sigact/ChrobakK06} for further pointers).

\section{Upper Bounds} \label{sec:upperbound} 
\noindent We provide upper bounds for the isomorphism problem by developing appropriate algorithms. Of course by querying the oracle for the input trees alternatingly, there is an obvious deterministic algorithm with a complexity of~$O(\min\{{|T_1|}, {|T_2|}\})$.
For randomized algorithms, we start with a simple Monte Carlo algorithm.
Subsequently we argue that there is also a Las Vegas algorithm, (i.e., an algorithm that always answers correctly) which still has a good expected runtime if there is a modest bound on the maximum degree of the tree.

\subsection{Probabilistic Bidirectional Search} \label{sec:upperbound:probabilistic}

\begin{algorithm}[t] \label{alg:random_walk}
	\SetAlgoLined
	\SetAlgoNoEnd
	\caption[Random Walk of the Search Tree]{Random Walk of the Search Tree}
	\Fn{\RandomWalk{$v$}}{
		\SetKwInOut{Input}{Input}
		\SetKwInOut{Output}{Output}
		\Input{vertex $v$ of a black box search tree}
		\Output{a random leaf of the search tree rooted at $v$}
		\While{$\deg(v) \neq 0$}{
			$v$ := \RandomChild{$v$}\;
		}
		\Return{$v$}\;
	}
\end{algorithm}

The central idea of the probabilistic isomorphism test discussed in this section is to perform \emph{random root to leaf walks} in the black box search trees. 
The recursive procedure for such walks is described in Algorithm~\ref{alg:random_walk} and simply works as follows:
a random walk is performed by starting in the root node and repeatedly choosing uniformly at random a child of the current node, until a leaf is reached.

Repeatedly executing random walks, the probabilistic isomorphism test exploits the following observation: assume we have two isomorphic trees $T_1$, $T_2$. 
Further assume we fix some leaf $l \in T_1$. We call all leaves $l'$ with $\Col(l) = \Col(l')$ \emph{occurrences} of $l$. 
The algorithm tries to find occurrences of $l$ through random walks of the trees. 
Towards finding $l$, we always perform two random walks, one in $T_1$ and one in $T_2$. 
Since we assumed the trees \emph{are isomorphic}, we are \emph{equally likely} to find an occurrence of $l$ in $T_1$ or in $T_2$. But if the trees \emph{are not isomorphic}, we can find occurrences of $l$ \emph{only} in $T_1$ (otherwise, due to the isomorphism invariance axiom, $T_1$ and $T_2$ would be isomorphic).

Algorithm~\ref{alg:random_iso} describes a procedure based on this observation. 
Instead of using just a single leaf $l$ however, it uses two sets of leaves $L_1$ and $L_2$ for comparison. 
Whenever an entirely new leaf is found (that is not an occurrence of a previously found leaf), it is added to the respective set of leaves and used for subsequent testing. 

If a leaf is an occurrence of a previously discovered leaf, it either reveals an isomorphism between the trees or an automorphism (possibly the identity) of one of the trees. This is again a consequence of the isomorphism invariance axiom.
If an isomorphism\footnote{Slightly abusing terminology we often use the term isomorphism to mean an isomorphism between the two trees rather than an isomorphism that is an automorphism of one tree.} is discovered, the algorithm has found two equally colored leaves in both trees and terminates. 
Otherwise, after a certain number of automorphisms have been found (the number depends on the desired error bound), the algorithm concludes that the trees are probably non-isomorphic within the given error bound. 
If the trees are isomorphic, we find automorphisms and isomorphisms with equal probability. Hence, we are highly unlikely to discover many automorphisms without also discovering an isomorphism. Figure~\ref{fig:algorithm} illustrates this key concept underlying the algorithm.

The following lemma proves correctness of the algorithm. 

\colorlet{blueish}{cyan!70!green}
\colorlet{blueish1}{blueish!20}
\colorlet{blueish2}{blueish!30}
\colorlet{blueish3}{blueish!40}
\colorlet{blueish4}{blueish!50}
\colorlet{blueish5}{blueish!60}
\colorlet{blueish6}{blueish!70}
\colorlet{blueish7}{blueish!80}

\colorlet{yellowish}{yellow}
\colorlet{yellowish1}{yellowish!20}
\colorlet{yellowish2}{yellowish!30}
\colorlet{yellowish3}{yellowish!40}
\colorlet{yellowish4}{yellowish!50}
\colorlet{yellowish5}{yellowish!60}
\colorlet{yellowish6}{yellowish!70}
\colorlet{yellowish7}{yellowish!90}

\begin{figure}
	\centering
	\begin{tikzpicture}[scale=0.66]
		\node (r0) at ( 0.0,  0.0) {}; 
		\node (s0) at (-3.0, -4.0) {}; 
		\node (s1) at ( 3.0, -4.0) {}; 
		\node (si1) at (-2.0, -4.0) {};
		\node (si2) at (-1.5, -4.0) {};
		\node (si3) at ( 0, -4.0) {}; 
		\node (si4) at ( 5.5, -4.0) {}; 
		\node (si5) at ( 7.0, -4.0) {}; 
		\node (si6) at ( 9.0, -4.0) {}; 
		
		\fill[fill=gray!20,draw] (r0.center)--(s0.center)--(s1.center)--(r0.center);
		
		\node (T1) at  (-2, 0) {$T_1$};
		
		\node (r01) at ( 0.0 + 7,  0.0) {}; 
		\node (s01) at (-3.0 + 7, -4.0) {}; 
		\node (s11) at ( 3.0 + 7, -4.0) {}; 
		\fill[fill=gray!20,draw] (r01.center)--(s01.center)--(s11.center)--(r01.center);
		
		\node (T2) at  (7 + 2, 0) {$T_2$};
		
		\draw[color=black, line width=1.5pt,] (si1) to [out=290, in=250] (si3) node [label=below:{\tiny \hspace{0.1cm} auto}] {}; 
		\draw[color=black, line width=1.5pt, ] (si1) to [out=290, in=250] (si4) node [label=below:{\tiny \hspace{0.2cm} iso}] {}; 
		
		\draw[color=black, line width=1.5pt,densely dotted] (r0) to [out=250, in=90] (si1); 
		\draw[color=black, line width=1.5pt,densely dotted] (r0) to [out=270, in=90] (si2); 
		\draw[color=black, line width=1.5pt,densely dotted] (r0) to [out=290, in=90] (si3); 
		
		\draw[color=black, line width=1.5pt,densely dotted] (r01) to [out=250, in=90] (si4); 
		\draw[color=black, line width=1.5pt,densely dotted] (r01) to [out=270, in=90] (si5); 
		\draw[color=black, line width=1.5pt,densely dotted] (r01) to [out=290, in=90] (si6); 
		
		\draw[color=black, fill=white]  (r01) circle (.15);
		\draw[color=black, fill=white]  (r0) circle (.15);
		
		\draw[color=black, fill=orange]  (si1) circle (.15);
		\draw[color=black, fill=blueish1]  (si2) circle (.15);
		\draw[color=black, fill=orange]  (si3) circle (.15);
		\draw[color=black, fill=orange]  (si4) circle (.15);
		\draw[color=black, fill=yellowish2]  (si5) circle (.15);
		\draw[color=black, fill=yellowish7]  (si6) circle (.15);
	\end{tikzpicture}
	\caption{The probabilistic bidirectional search algorithm simultaneously samples leaves in both trees using random walks. It then tests for automorphisms within a tree and isomorphisms across trees to perform the probabilistic test.} \label{fig:algorithm}
\end{figure}
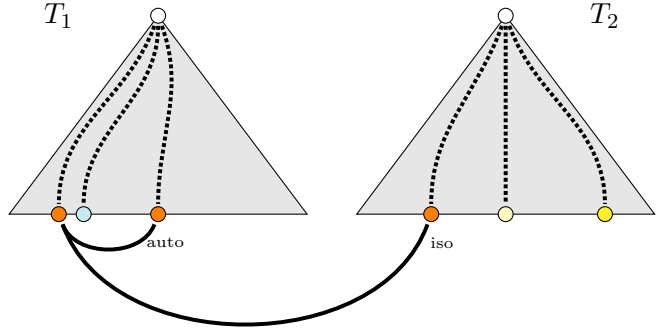

\begin{algorithm}[t!] \label{alg:random_iso}
	\SetAlgoLined
	\SetAlgoNoEnd
	\caption[Probabilistic Bidirectional Search]{Probabilistic Bidirectional Search}
	\Fn{\RandomIso{$T_1$, $T_2$, $\epsilon$}}{
		\SetKwInOut{Input}{Input}
		\SetKwInOut{Output}{Output}
		\Input{black box search trees $T_1, T_2$ and probability $\epsilon$}
		\Output{two leaves $l_1 \in T_1$, $l_2 \in T_2$ such that $\Col(l_1) = \Col(l_2)$ with probability at least $1 - \epsilon$ if such leaves exist, $\bot$ otherwise}
		$c$    := $0$\;
		$e$    := $\ceil{-\log_2(\epsilon)}$\;
		$L_1 := L_2 := \emptyset$\;
		\While{$c \leq e$}{
			$f_{(aut, 1)}$ := $f_{(aut, 2)}$ := \texttt{false}\tcp*{$f_{(aut, i)}$ indicates automorphism found in $T_i$}
			$l_1$ := \RandomWalk{root of $T_1$} \label{alg:random_iso:walk1}\;
			$l_2$ := \RandomWalk{root of $T_2$} \label{alg:random_iso:walk2}\;
			\lIf{$\Col(l_1) = \Col(l_2)$}{\Return{$(l_1, l_2)$}}
			\For{$i \in \{1, 2\}$}{
				\For{$l' \in L_{(3-i)}$}{
					\lIf{$\Col(l_i) = \Col(l')$}{
						\Return{$(l_i, l')$} 
					}
					\lIf{$\Col(l_{3-i}) = \Col(l')$}{
						$f_{(aut, (3-i))}$ := \texttt{true}
					}
				}
			}
			\lIf{$\neg f_{(aut, 1)}$}{$L_1 := L_1 \cup \{l_1\}$}
			\lIf{$\neg f_{(aut, 2)}$}{$L_2 := L_2 \cup \{l_2\}$}
			\lIf{$f_{(aut, 1)} \vee f_{(aut, 2)}$}{$c := c+1$}
		}
		\Return{$\bot$}\;
	}
\end{algorithm}

\begin{lemma} Given black box search trees $T_1, T_2$ and a desired error probability $\epsilon$, Algorithm~\ref{alg:random_iso} solves the isomorphism exploration problem
with probability at least $1 - \epsilon$.
\end{lemma}
\begin{proof} 
	First, observe that whenever a pair of leaves is returned their color is checked for equality. This ensures that if the algorithm returns a pair of leaves, the answer is always correct. The algorithm can therefore only fail to produce the correct output by not finding a suitable pair of equally colored leaves despite the fact that they exist. In particular, this implies that if the trees are non-isomorphic, the algorithm cannot err.

	To bound the error probability, we view the computation as a sequence of \emph{tests}. 
	A test repeatedly performs random walks of the search trees until one automorphism (possibly the identity) or one isomorphism is found. 
	Hence, each test can be described as a sequence of $j$ iterations. In each iteration $j' < j$, neither $l_1$ nor $l_2$ produced an isomorphism or automorphism. 
	During a test, the algorithm neither terminates, nor is $c$ incremented. 
	In iteration $j$ of the test, an automorphism or isomorphism is found. 
	Now, note that when $T_1$ and $T_2$ are isomorphic, leaves contained in $L_i$ can equally likely be found in $T_1$ or $T_2$. 
	Hence, finding an automorphism or isomorphism in a test is equally likely. In particular, the probability is $\frac{1}{2}$ for finding an isomorphism for each~$i\in \{1,2\}$ rather than an automorphism. 
	Anytime we find an automorphism but no isomorphism, we increment $c$ by $1$. We terminate when $c$ reaches $e$. 
	Assuming the trees are isomorphic, the probability of this outcome is therefore bounded by $(\frac{1}{2})^e$.
\end{proof}
We now analyze the worst-case runtime of Algorithm~\ref{alg:random_iso} in the cost model of black box search trees.
This means that only Line~\ref{alg:random_iso:walk1} and Line~\ref{alg:random_iso:walk2} affect cost in the analysis.  
Since termination of the algorithm depends on random events the running time of the algorithm is a random variable and we consider expected runtime.

We are now ready to analyze the running time of Algorithm~\ref{alg:random_iso}.

\begin{lemma} Given black box search trees $T_1$, $T_2$ of heights~$h_1$ and~$h_2$, respectively, and error probability $\epsilon$, Algorithm~\ref{alg:random_iso} has an expected worst-case runtime bounded by
	\[
		\mathcal{O}\left(\ceil{\log_2(\frac{1}{\epsilon})}\cdot \max(h_1,h_2) \cdot\min\{\sqrt{|T_1|}, \sqrt{|T_2|}\}\right).
	\]
\end{lemma} 
\begin{proof}
We will calculate the expected number of leaves explored before termination. 
We may consider the number of leaves instead of nodes by adding the multiplicative factor $\max(h_1,h_2)$ for the maximum length of a root to leaf walk in the search trees to our runtime.
(We explain subsequently how to improve this to factor.)

We may assume that the input trees are non-isomorphic and thus that the algorithm terminates because the condition $c > e = \ceil{-\log_2(\epsilon)}$ was met. This suffices to give an upper bound since earlier termination due to the discovery of isomorphisms clearly only leads to a smaller expected running time.

Consider running $2\sqrt{|T_i|}$ iterations of the algorithm. 
We may assume that in the $j$-th iteration $L_1$ and $L_2$ each contain at least $j$ leaves:
otherwise, some previous iteration already discovered an automorphism or an isomorphism. 
Furthermore, we may assume that the probability to find a leaf is uniform across all leaves: if probabilities are non-uniform, the chance for finding some leaves repeatedly only increases (see~\cite{MR451317}). 
The probability of finding an automorphism in $L_i$ (with $i \in \{1, 2\}$) within $j$ iterations is therefore at least $\frac{j}{|T_i|}$.
After $\sqrt{|T_i|}$ iterations, the probability for finding an automorphism in $T_i$ is then at least 
\[
	\frac{\sqrt{|T_i|}}{|T_i|} = \frac{1}{\sqrt{|T_i|}}.
\]  
Hence, the probability of finding no automorphism after $2\sqrt{|T_i|}$ many steps is at most
\[
	\left(\frac{\sqrt{|T_i|}-1}{\sqrt{|T_i|}}\right)^{\sqrt{|T_i|}} \leq \frac{1}{\mathbf{e}}<\frac{1}{2}.
\]
We view the computation as a series of batches consisting of $2\sqrt{|T_i|}$ iterations each. For each of them, the probability for finding an automorphism is at least $1/2$.
For termination, we need to find $e$ many automorphisms. The expected number of batches is thus in~$\mathcal{O}(e)$, which shows that the overall number of iterations is in~$\mathcal{O}(e\cdot  2\sqrt{|T_i|})$.
\end{proof}


We can improve the bound on the running time replacing the factor $\max\{h_1,h_2\}$ with the factor $\log_2(\min\{\sqrt{|T_1|}, \sqrt{|T_2|}\})$. To do so we alter the algorithm to take into account that the trees may be of very different sizes and also the trees may be quite unbalanced. To compensate for this we employ a doubling technique. However, we first need a bound for the expected length of the random root to leaf walks used in our algorithm.

\begin{lemma}\label{lem:exp:height}
In an~$n$-node black box search tree the expected length of a random root to leaf walk (i.e., the running time of Algorithm~\ref{alg:random_walk}) is in~$O(\log n)$.
\end{lemma}
\begin{proof}
Let~$g(T)$ be the expected length of a random root to leaf walk in tree~$T$. 
Note that the number of leaves~$t$ of a black box search tree is in~$\Theta(n)$ for an~$n$-node tree.
We will argue that among the trees~$T$  with~$t$ leaves the value~$g(T)$ is maximal if~$T$ is a binary tree in which all leaves are located on two consecutive levels. Since in such a tree even the maximum root to leaf distance is~$O(\log n)$, this proves the theorem.

First let~$T$ be a tree which has a vertex~$v$ with more than two children~$u_1,\ldots,u_j$. Let~$T_i$ be the subtree of~$T$ rooted at~$u_i$ and assume without loss of generality that~$g(T_i)<g(T_j)$ for~$i<j$.
Alter the tree~$T$ into a new tree~$T'$ by inserting a new node~$w$ as a child of~$v$ and then relocating the trees~$T_1$ and~$T_2$ so that their roots are now children of~$w$ instead of~$v$.
Then, conditional on the event that the random walk reaches~$v$ the expected length of the walk has increased.
Thus~$g(T')>g(T)$. 
Since there are only finitely many trees with~$t$ leaves, by induction it suffices now to consider binary trees.

Let~$T$ be a binary tree and suppose there are leaves~$\ell_1$ and~$\ell_2$ whose height differs by more than~1. Say~$\ell_1$ is on the level furthest from the root. There must be another leaf~$\ell_3$ whose parent~$p$ is also the parent of~$\ell_1$.
Alter the tree to obtain a new tree~$T'$ by assigning~$\ell_2$ as the new parent of~$\ell_3$ and~$\ell_2$. 
Note that~$p$ is further away from the root than~$\ell_2$. Thus, the tree being binary, the probability of a random walk reaching~$\ell_2$ is larger than that of reaching~$p$. Therefore~$g(T')>g(T)$. By induction this proves the theorem.
\end{proof}

\begin{theorem}
There is an algorithm that solves the isomorphism exploration problem with probability at least $1 - \epsilon$ and expected worst-case runtime bounded by
\[
		\mathcal{O}\left(\ceil{\log_2(\frac{1}{\epsilon})}\cdot \log_2(\min\{\sqrt{|T_1|}, \sqrt{|T_2|}\}) \cdot\min\{\sqrt{|T_1|}, \sqrt{|T_2|}\}\right).
	\]
\end{theorem}

\begin{proof}
Set~$n= \min\{|T_1|, |T_2|\}$. 
For an integer~$s$, we run the algorithm with a budget $2s$ that limits the number of walks that can be performed in each tree to~$s$.
Furthermore, we limit the length of the random walks by $h = c\log_2(s)$ for some suitable constant determined later.
Whenever a random walk exceeds the length $h$, we abort the walk and ignore it. 
If the algorithm does not terminate within the alloted budget then we double~$s$ and restart.
This guarantees that the number of queries does not exceed~$O(s\log s)$ when we run it with integer~$s$.

At least in the smaller of the two trees, 
automorphisms are found with high probability whenever $s$ exceeds $\sqrt{n}$. Indeed, by Lemma~\ref{lem:exp:height} the average length of a random walk in the smaller tree is in~$O(\log n)= O(\log \sqrt n)$. Thus, by Markov's bound with probability~$1/2$, the random walks end in a leaf of height at most~$O(\log n)$. 
Thus, by the Chernoff bound, for sufficiently large~$s$, with probability~$1/2$ at least~$1/4$ of the random walks end in a leaf of height at most~$O(\log n)$. We choose~$c$ so that this height is at most~$c \log_2 (n)$.

In case the graphs are isomorphic, automorphisms and isomorphisms are still found with equal probability. Thus our arguments for the probabilities remain in place since we essentially perform the same algorithm in pruned subtrees.
Regarding the running time, note that the probability that the algorithm does not terminate with budget~$s$ decreases exponentially with~$s$. That is, the probability is in~$O(a^{s/\min\{\sqrt{|T_1|}, \sqrt{|T_2|}\}})$ for some constant~$a<1$ once~$s>2\min\{\sqrt{|T_1|}, \sqrt{|T_2|}\}$.
\end{proof}

We should remark that the collision problem was previously exploited in the context of the group isomorphism problem~\cite{DBLP:conf/soda/Rosenbaum13}, but there structural information on the corresponding trees is known. Also the idea of sampling with random walks was used for isomorphism algorithm in~\cite{thesis}, but that algorithm only uses a single leaf in the search tree and thus cannot achieve sublinear running time guarantees.

\subsection{Las Vegas Bidirectional Search} \label{sec:upperbound:lasvegas} 

The major drawback of the bidirectional search algorithm is that it makes errors. 
Considering trees of height 1 it is not difficult to see that a non-erring algorithm, even a randomized one, will need to query a linear fraction of the leaves to distinguish non-isomorphic trees. 
However, if the degree of the input graphs is restricted, we can beat this bound.

To do this, we basically strive to choose a specific set of nodes in both trees that ensures a ``collision'' of leaves. This guarantees that we find equally colored leaves, if they exist.

We refer to the maximum degree among the considered trees as $d$.
The main new idea we use is to split the search tree in a balanced manner, followed by techniques to exploit isomorphism invariance. 
We want to note that the techniques for exploiting isomorphism invariance are inspired by the techniques described in \cite{McKay201494, DBLP:journals/jea/Stoichev19}, which essentially also perform splits. However, rather than heuristically applying them, here we perform them in a balanced and systematic way. 
Towards this goal we need the notion of a split $(v, h)$, which is simply a node $v \in T_i$ at level $h$ in one of the input trees. 
We define the \emph{cost} of a split as a pair of numbers $(s_1, s_2)$ as follows: 
\begin{enumerate} 
	\item $s_1$ is the size of the tree $T_{3-i}$ truncated at level~$h$ (i.e., the ball of radius $h$ around the root).
	\item If the tree $T_{i}$ truncated at level~$h$ is non-isomorphic to the tree $T_{3-i}$ truncated at level~$h$ then $s_2 := s_1$, \emph{otherwise} $s_2$ is the size of the subtree rooted in $v \in T_{i}$ at level~$h$. 
\end{enumerate}
The intuition for our exploration strategy is that $s_1$ bounds the size of the subtree to be explored in $T_{3-i}$, while $s_2$ bounds the size of the subtree to be explored in $T_i$ (up to logarithmic factors). 
While the definition for (2) may seem cumbersome at first, the idea is simply that if trees already differ in the first $h$ levels, we can decide non-isomorphism by exploring all nodes in the subtree of~$T_{i}$ consisting of the first~$h$ levels and then at most as many vertices within the first~$h$ levels of~$T_{3-i}$.

We call a split~$(v, h)$ a \emph{balanced split} whenever its cost $(s_1, s_2)$ satisfies $\max\{s_1, s_2\} \leq 4d \cdot \min\{\sqrt{|T_1|}, \sqrt{|T_2|}\}$.
Note that slightly abusing terminology in a balanced split the subtree with root~$v$ can be large, as long as the two trees truncated at level~$h$ are non-isomorphic and~$s_1$ is sufficiently small.

At this point it might neither be clear how to find a balanced split nor that a balanced split always exists. However, assume for now that we are given a balanced split. In that case we can efficiently solve isomorphism (even deterministically) as follows.
We perform breadth-first search up to level $h$ in both trees $T_1$ and $T_2$, visiting all nodes $N_1 \subseteq V(T_1)$ and $N_2 \subseteq V(T_2)$ up to and including level $h$. 
We can conclude non-isomorphism immediately whenever the breadth-first search has finished level $h$ and the two trees truncated at level~$h$ are non-isomorphic. 
We can thus assume now that these trees are isomorphic. 
By exploring all nodes up to level~$h$ (which is the level containing $v$), we surely explore the node $v$ in one of the trees. 
Without loss of generality assume in the following that $v \in V(T_1)$. 

In $T_1$, we explore \emph{all} leaves $L_v$ of the subtree rooted at \emph{one} fixed node, namely $v$ from the balanced split.
Let $N_2' \subseteq N_2$ denote the set of nodes at level $h$ in~$T_2$. 
Then, we explore for each node $v'$ in $N_2'$ \emph{one} arbitrary leaf $l_{v'}$ in the subtree rooted at~$v'$. 
If the trees are isomorphic, there must exist some $v' \in N_2'$ that can be mapped to $v$ with an isomorphism. Since we explored all leaves of $v$, the leaf $l_{v'}$ with ancestor $v'$ must be isomorphic (equally colored) to one of the leaves in $L_v$. 
Figure~\ref{fig:deterministic_iso} illustrates how the collision of leaves is enforced through the exploration strategy. 

The procedure for exploring, within our model, the first~$h$ levels of the subtree rooted at a particular node~$v$ is described in Algorithm~\ref{alg:subtree}.  
Starting from a given node $v$, it performs breadth-first traversal until only leaves are left, or $h$ levels have been explored. The algorithm is also given a cost limit~$s$ and the algorithm aborts if this limit is reached.

Using Algorithm~\ref{alg:subtree} as a subroutine, Algorithm~\ref{alg:deterministic_iso} gives an implementation in the exploration model of the entire algorithm just described.

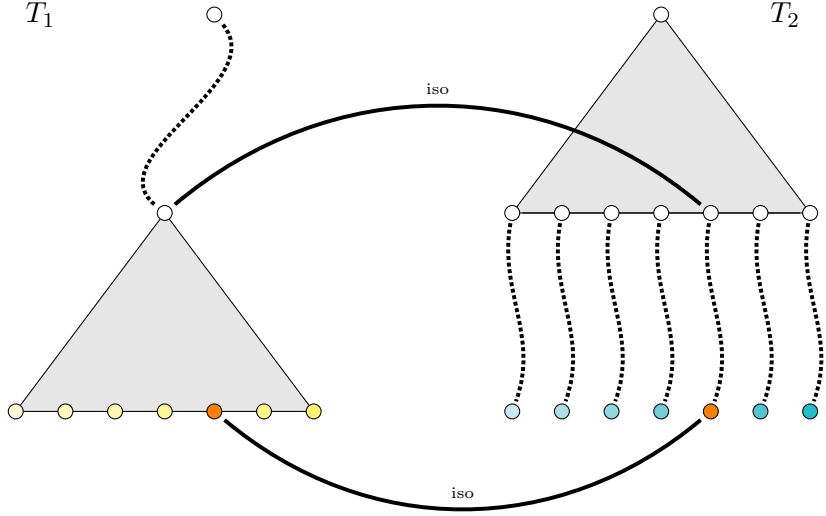
\begin{figure}
	\centering
	\begin{tabular}{c c c}
		\begin{tikzpicture}[scale=0.66]
			\node (lr0) at ( 0.0-9.0,  0.0) {};
			\node (ls0) at (-3.0-9.0, -4.0) {}; 
			\node (ls1) at ( 3.0-9.0, -4.0) {};
			\node (lsi) at (-1.0-9.0, -4.0) {};
			\node (lsi0) at (-1.0 - 3-9.0, -4.0 - 4) {};
			\node (lsi1) at (-1.0 + 3-9.0, -4.0 - 4) {};
			\node (la) at (-2.5 -1-9.0, 0) {$T_1$ };

			\node (ll0) at (-3.0 - 1-9.0, -8.0) {}; 
			\node (ll1) at ( 3.0 - 1-9.0, -8.0) {}; 
			\node (ll2) at (-2.0 - 1-9.0, -8.0) {}; 
			\node (ll3) at (-1.0 - 1-9.0, -8.0) {};
			\node (ll4) at (0.0 - 1-9.0, -8.0) {};
			\node (ll5) at (1.0 - 1-9.0, -8.0) {};
			\node (ll6) at (2.0 - 1-9.0, -8.0) {};

			\fill[fill=gray!20,draw] (lsi.center)--(lsi0.center)--(lsi1.center)--(lsi.center);

			\draw[color=black, fill=white] (lr0) circle (.15);
			\draw[color=black, fill=white]   (lsi) circle (.15);
			\draw[color=black, fill=yellowish1]  (ll0) circle (.15);
			\draw[color=black, fill=yellowish6]  (ll1) circle (.15);
			\draw[color=black, fill=yellowish2]  (ll2) circle (.15);
			\draw[color=black, fill=yellowish3]  (ll3) circle (.15);
			\draw[color=black, fill=yellowish4]  (ll4) circle (.15);
			\draw[color=black, fill=orange]  (ll5) circle (.15);
			\draw[color=black, fill=yellowish5]  (ll6) circle (.15);

			\draw[color=black, line width=1.5pt,densely dotted] 
			(lr0) to [out=310, in=140] (lsi); 
			\node (a) at (2.5, 0) {$T_2$ };

			\node (rr0) at ( 0.0,  0.0) {}; 
			\node (s0) at (-3.0, -4.0) {}; 
			\node (s1) at ( 3.0, -4.0) {}; 
			\node (s2) at (-2.0, -4.0) {}; 
			\node (s3) at (-1.0, -4.0) {}; 
			\node (s4) at (0.0, -4.0) {}; 
			\node (s5) at (1.0, -4.0) {};
			\node (s6) at (2.0, -4.0) {};

			\node (l0) at (-3.0, -8.0) {}; 
			\node (l1) at ( 3.0, -8.0) {}; 
			\node (l2) at (-2.0, -8.0) {}; 
			\node (l3) at (-1.0, -8.0) {}; 
			\node (l4) at (0.0, -8.0) {}; 
			\node (l5) at (1.0, -8.0) {};
			\node (l6) at (2.0, -8.0) {};

			\fill[fill=gray!20,draw] (rr0.center)--(s0.center)--(s1.center)--(rr0.center);
			\path[draw] (s0)--(s1);
			\draw[color=black, fill=white] (rr0) circle (.15);

			\draw[color=black, line width=1.5pt,densely dotted] (s0) to [out=260, in=70] (l0); 
			\draw[color=black, line width=1.5pt,densely dotted] (s1) to [out=260, in=70] (l1); 
			\draw[color=black, line width=1.5pt,densely dotted] (s2) to [out=260, in=70] (l2); 
			\draw[color=black, line width=1.5pt,densely dotted] (s3) to [out=260, in=70] (l3); 
			\draw[color=black, line width=1.5pt,densely dotted] (s4) to [out=260, in=70] (l4); 
			\draw[color=black, line width=1.5pt,densely dotted] (s5) to [out=260, in=70] (l5); 
			\draw[color=black, line width=1.5pt,densely dotted] (s6) to [out=260, in=70] (l6); 
			\draw[color=black, fill=white]  (s0) circle (.15);
			\draw[color=black, fill=white]  (s1) circle (.15);
			\draw[color=black, fill=white]  (s2) circle (.15);
			\draw[color=black, fill=white]  (s3) circle (.15);
			\draw[color=black, fill=white]  (s4) circle (.15);
			\draw[color=black, fill=white]  (s5) circle (.15);
			\draw[color=black, fill=white]  (s6) circle (.15);

			\draw[color=black, fill=blueish1]  (l0) circle (.15);
			\draw[color=black, fill=blueish6]  (l1) circle (.15);
			\draw[color=black, fill=blueish2]  (l2) circle (.15);
			\draw[color=black, fill=blueish3]  (l3) circle (.15);
			\draw[color=black, fill=blueish4]  (l4) circle (.15);
			\draw[color=black, fill=orange]  (l5) circle (.15);
			\draw[color=black, fill=blueish5]  (l6) circle (.15);

			\draw[color=black, line width=1.5pt] (lsi) to [out=40, in=140] node [above, pos=0.5] {\tiny iso} (s5);
			\draw[color=black, line width=1.5pt] (ll5) to [out=320, in=220] node [above, pos=0.5] {\tiny iso} (l5);

		\end{tikzpicture}
	\end{tabular}
	\caption{State of the search trees after termination of Algorithm~\ref{alg:deterministic_iso}. If trees are isomorphic, a node $v_1$ in $T_1$ at some level $h$ must be mapped to some node $v_2$ in $T_2$ at level $h$ by an isomorphism. But if that is the case, then a leaf below $v_2$ is isomorphic to some leaf below $v_1$.} \label{fig:deterministic_iso}
\end{figure}

\begin{algorithm}[t] \label{alg:subtree}
	\SetAlgoLined
	\SetAlgoNoEnd
	\caption[Compute a Subtree of the Search Tree]{Breadth-first Search: Computing a Subtree of the Search Tree}
	\Fn{\Subtree{$v, h, s$}}{
		\SetKwInOut{Input}{Input}
		\SetKwInOut{Output}{Output}
		\Input{start node $v$, height limit $h$, cost limit $s$}
		\Output{$(N, s')$ where $N$ is the set of leaves of the subtree under $v$ up to level $h$ and $s'$ is the number of explored nodes, or $(\bot, \bot)$ if cost limit did not suffice}
		\lIf{$h = 0$}{\Return{$\{v\}$}}
		$N$      := $\{(v, 0)\}$\;
		$L$      := $\{\}$\;
		$s'$     := 0\;
		\While{$N \neq \emptyset$}{
			$(v, h') := $ \Some{N} \tcp*{pick arbitrary element of $N$}
			$N$     := $N \smallsetminus \{(v, h')\}$\;
			$h' := h' + 1$\;
			$c  := \NextChild(v)$\;
			\While{$c \neq \bot$}{
				$s' := s' + 1$\;
				\lIf{$s' > s$}{
					\Return{$(\bot, \bot)$}
				}
				\lIf{$h' = h \vee \deg(c) = 0$}{
					$L$ := $L \cup \{c\}$
				}
				\lElse{
					$N$ := $N \cup \{(c, h')\}$
				}
				$c  := \NextChild(v)$\;
			}
		}
		\Return{$(L, s')$}\;
	}
\end{algorithm}

\begin{algorithm} \label{alg:deterministic_iso}
	\SetAlgoLined
	\SetAlgoNoEnd
	\caption[Bidirectional Search]{Bidirectional Search}
	\Fn{\RandomIso{$T_1, T_2, v, h$}}{
		\SetKwInOut{Input}{Input}
		\SetKwInOut{Output}{Output}
		\Input{black box search trees $T_1, T_2$ and a split $v, h$ with $v \in V(T_1)$}
		\Output{two leaves $l_1 \in T_1$, $l_2 \in T_2$ such that $\Col(l_1) = \Col(l_2)$ if they exist, $\bot$ otherwise}
		$(N_2, s)\;$   := \Subtree{root of $T_2, h, \infty$} \label{alg:deterministic_iso:n2}\;
		$(N_1, \_)$  := \Subtree{root of $T_1, h, s$} \label{alg:deterministic_iso:n1} \tcp*{explores $v$}
		\If{$N_1 = \bot$ or $T_1$ and $T_2$ up to level $h$ non-isomorphic}{
			\Return{$\bot$}\;
		}
		$(L_v, \_)$   := \Subtree{$v, \infty, \infty$}\;
		\For{$n \in N_2$}{ \label{alg:deterministic_iso:walks}
			$l :=$ \RandomWalk{$n$}\tcp*{can be an arbitrary, even a determinsitic walk}
			\For{$l' \in L_v$}{
				\lIf{$\Col(l) = \Col(l')$}{
					\Return{$(l, l')$}
				}
			}
		}
		\Return{$\bot$}\;
	}
\end{algorithm} 

Let us argue an upper bound for the runtime of Algorithm~\ref{alg:deterministic_iso} (still assuming that we are given a balanced split). 
From the definition of a balanced split, we can conclude that $|L_v|$ and $|N_2|$ are bounded by $\mathcal{O}(d \cdot \min\{\sqrt{|T_1|}, \sqrt{|T_2|}\})$. 
Since exploration up to level $h$ in $T_1$ (Line~\ref{alg:deterministic_iso:n1}) may only explore as many nodes as exploration in $T_2$, we ensure that $|N_1| \leq |N_2|$ holds.
Now, the last phase probes at most $\mathcal{O}(d \cdot \min\{\sqrt{|T_1|}, \sqrt{|T_2|}\})$ many paths, giving an overall upper bound of
\[\mathcal{O}\left(d \cdot h(T_2)\cdot \min\{\sqrt{|T_1|}, \sqrt{|T_2|}\}\right).\]

However, note that the factor $h(T_2)$ can be excessively large because the paths from level~$h$ to the leaves can be of length~$\Theta(|T_2|)$.
 To prevent this, we alter the algorithm as follows. In~$T_2$ we allocate a total budget of~$c' \sqrt{|T_2|} \log(|T_2|)$ for some constant~$c'$ for all the level-$h$-to-leaves paths.
By Lemma~\ref{lem:exp:height} the expected length of one such path is in~$O(\log |T_2|)$. Thus, by linearity of expectation, the expected total cost for the paths is~$O(\sqrt{|T_2|} \log |T_2|)$. By Markov's inequality for a suitable choice of~$c'$, with probability~$1/2$ the total cost is in~$O(\sqrt{|T_2|} \log |T_2|)$. If the total cost exceeds this bound we simply restart the process.

Overall we can replace the factor $h(T_2)$ by~$O(\log(|T_2|))$, giving
\[\mathcal{O}(d \cdot \log_2(\max\{|T_1|, |T_2|\}) \cdot \min\{\sqrt{|T_1|}, \sqrt{|T_2|}\}).\]
More generally, it is easy to see that given a split of cost $(s_1, s_2)$, the modification of Algorithm~\ref{alg:deterministic_iso} runs in
\[\mathcal{O}(\log_2(\max\{{|T_1|}, {|T_2|}\}) \cdot \max\{s_1, s_2\}).\]
There is an interesting analogy to the runtime of the probabilistic bidirectional search algorithm. A main difference is that the runtime directly depends on the maximum degree of the trees.

The crucial question remains whether balanced splits always exist and whether they can be found efficiently.
We first address the question of existence of balanced splits.
\begin{lemma} \label{lem:eq_split_existence}
	Let $T_1, T_2$ be black box search trees with maximum degree $d$. Then there exists a balanced split for search trees $T_1$ and $T_2$. 

	In particular, if 
	\begin{itemize}
		\item $h'$ is the maximal level for which the tree $T_1$ truncated at level~$h'$ is smaller than $4d \cdot \min\{\sqrt{|T_1|}, \sqrt{|T_2|}\}$,
		\item  the two subtrees up to level $h'$ are isomorphic, and 
		\item  there are no leaves up to level $h'$,
	\end{itemize}
	then at least $\frac{3}{4}$ of the nodes at level $h'$ in the smaller tree constitute balanced splits with cost $s_2 \leq 2 \min\{\sqrt{|T_1|}, \sqrt{|T_2|}\}$. 
\end{lemma} 
\begin{proof} 
	We can assume w.l.o.g. that $|T_1| \leq |T_2|$.
	Let $h'$ be the maximal level of $T_2$ where the size of the subtree up to level $h'$ is smaller than or equal to $4d \cdot \min\{\sqrt{|T_1|}, \sqrt{|T_2|}\}$. 

	If the subtrees up to level $h'$ in $T_1$ and $T_2$ differ, we have found a balanced split. 
	Furthermore, if there are leaves in the trees up to level $h'$, we have found a balanced split as well. 
	Hence, we assume that subtrees are isomorphic and no leaves are present. 

	We now argue that at least $\frac{3}{4}$ of the nodes at level $h'$ in $T_1$ constitute balanced splits.
	Consider level $h'$ of $T_1$ and $T_2$.  
	Let $s_{h'} \leq 4d \cdot \min\{\sqrt{|T_1|}, \sqrt{|T_2|}\}$ be the size of the subtree up to and including level $h'$ in $T_1$.
	By assumption, the respective subtree of $T_2$ is of equal size.
	Furthermore, by assumption there are no leaves up to level $h'$, implying that the tree contains at least $n_{h'} \geq \frac{1}{2}\cdot s_{h'}$ nodes at level $h'$.

	Towards a contradiction, we assume $s_{h'} \leq 4\cdot\min\{\sqrt{|T_1|}, \sqrt{|T_2|}\}$. But then we can increment~$h'$: since $s_{h'} \leq 4\cdot \min\{\sqrt{|T_1|}, \sqrt{|T_2|}\}$, it holds that $s_{h' + 1} \leq 4d \cdot \min\{\sqrt{|T_1|}, \sqrt{|T_2|}\}$. This is a contradiction to the assumption that $h'$ is maximal. Hence, we know $4 \cdot \min\{\sqrt{|T_1|}, \sqrt{|T_2|}\} < s_{h'} \leq 4d \cdot \min\{\sqrt{|T_1|}, \sqrt{|T_2|}\}$.

	We can immediately conclude $n_{h'} \geq 2 \cdot \min\{\sqrt{|T_1|}, \sqrt{|T_2|}\}$. Naturally, there can be at most $\frac{1}{2} \cdot \min\{\sqrt{|T_1|}, \sqrt{|T_2|}\}$ corresponding subtrees which have a size greater than $2 \cdot \min\{\sqrt{|T_1|}, \sqrt{|T_2|}\}$ with roots at level $h'$ in $T_1$ (since $|T_1| \leq |T_2|$). 
	Consequently, there must be at least $n_{h'} - \frac{1}{2} \cdot \min\{\sqrt{|T_1|}, \sqrt{|T_2|}\} \geq \frac{3}{4}n_{h'}$ subtrees rooted at level $h'$ with a size smaller than $2\cdot\min\{\sqrt{|T_1|}, \sqrt{|T_2|}\}$. This concludes the proof for both claims of the lemma. 
\end{proof}
\noindent Thus, for all search trees there exist balanced splits. 
However, we still need to explain how to find balanced splits efficiently. 
As shown by the lower bound in the next section, it is impossible to do this deterministically in an adequate running time. 
We thus need a randomized procedure for finding balanced splits.
We will show that the following method is suitable. Rather than pseudocode we give a high level description of the algorithm.
\refstepcounter{algocf}
\noindent \paragraph{Algorithm \arabic{algocf} \textnormal{(Las Vegas Balanced Splits)}.} \label{alg:lasvegas_bidirectional}\mbox{}\\
\makebox[1.5cm][l]{\emph{Input:}}  Black box search trees $T_1$ and $T_2$.\mbox{}\\
\makebox[1.5cm][l]{\emph{Output:}} The algorithm either returns a split $(v, h)$ or determines that $T_1$ and $T_2$ are non-isomorphic.
\begin{enumerate}
	\item Set cost limit $s \leftarrow 1$.
	\item \label{alg:lasvegas_bidirectional:bfs} Perform breadth-first search in $T_1$ and $T_2$, limiting the size of the traversed subtree to $s$ nodes (each). If after any level the breadth-first search trees for $T_1$ and $T_2$ differ, the algorithm terminates concluding non-isomorphism. Let $h$ denote the level reached so far. If breadth-first search discovers a leaf $v$ at or below level $h$, the algorithm terminates with the split $(v, h)$. 
	\item \label{alg:lasvegas_bidirectional:choose} For each $i \in \{1, 2\}$, uniformly and independently at random choose a node $v_i$ at level $h$ in both trees.
	\item Compute breadth-first search starting from the node $v_i$ in $T_i$, until one of the following conditions is met:
	\begin{itemize}
		\item Breadth-first search finishes exploring the entire subtree of $v_i$, constituting the split $(v_i, h)$.
		\item Breadth-first search explored $s$ nodes.
	\end{itemize}
	This step is performed in parallel for both $i \in \{1, 2\}$ (i.e., each step alternates between the two), until one method succeeds in finding a split or both finish unsuccessfully. If at any point a split is found, the algorithm terminates immediately returning the split.
	\item Set $s \leftarrow 2s$ and jump to Step~\ref{alg:lasvegas_bidirectional:bfs}.
\end{enumerate}
It turns out that when the cost limit is sufficiently high the algorithm finds balanced splits with good probability. The intuition behind this is based on Lemma~\ref{lem:eq_split_existence}: once the algorithm reaches a sufficiently high level of the tree with breadth-first search, the majority of nodes constitute balanced splits. 

From the description of the algorithm, the following corollary follows readily:
\begin{corollary} \label{lem:eq_split_leaf}
If Algorithm~\ref{alg:lasvegas_bidirectional} chooses in some iteration an element $v$ at some level~$h$ in Step~\ref{alg:lasvegas_bidirectional:choose} so that~$(v,h)$ constitutes a split with cost $(s_1, s_2)$ and at this point $s \geq s_2$, then the algorithm terminates after this iteration and returns a split with cost $(s_1', s_2')$ where $s_1' = s_1, s_2' \leq s_2$.
\end{corollary}
\begin{proof} By assumption, $(v, h)$ constitutes a split with cost $(s_1, s_2)$. This implies that the entire subtree below $v$ is smaller than $s_2$. Consequently, since $s$ is large enough, Algorithm~\ref{alg:lasvegas_bidirectional} explores the entire subtree below $v$ in Step~4, unless probing the other way in parallel terminates first: this only contradicts our claim if it results in a more expensive split. However, since a more costly split necessitates more steps to explore its respective subtree, running the search in parallel ensures that the cheaper split is found first. Note that $s_1 = s_1'$ holds for all nodes at level $h$, concluding the proof.
\end{proof}
Using this, we can prove that Algorithm~\ref{alg:lasvegas_bidirectional} terminates with a balanced split with good probability, giving what we need for our isomorphism test:
\begin{lemma} \label{lem:eq_split_compute}
If Algorithm~\ref{alg:lasvegas_bidirectional} terminates with a split, it constitutes a balanced split with a probability of at least $\frac{3}{4}$.
\end{lemma}
\begin{proof} 
We can assume w.l.o.g.\ that $|T_1| \leq |T_2|$.
Let $h'$ be the maximal level of $T_2$ where the size of the subtree up to level $h'$ is smaller than or equal to $4d \cdot \min\{\sqrt{|T_1|}, \sqrt{|T_2|}\}$. 

First, we observe that we may always assume that the breadth-first trees explored in $T_1$ and $T_2$ up to level $h'$ are isomorphic, since otherwise Algorithm~\ref{alg:lasvegas_bidirectional} terminates immediately with no split (Step~\ref{alg:lasvegas_bidirectional:bfs}).
Furthermore, since Algorithm~\ref{alg:lasvegas_bidirectional} terminates when discovering leaves within the first $h'$ levels in the breadth-first exploration (Step~\ref{alg:lasvegas_bidirectional:bfs}) 
and this result in balanced splits, we may assume that Algorithm~\ref{alg:lasvegas_bidirectional} explores no leaves in the breadth-first search.  
We note that if Algorithm~\ref{alg:lasvegas_bidirectional} finds no leaves, each doubling of $s$ can only increase the level $h$ reached by breadth-first search by at most~$1$.  

Consider now Step~\ref{alg:lasvegas_bidirectional:choose} in the algorithm once level $h'$ is reached.
The algorithm picks a node of level $h'$ uniformly at random.
We now argue that with probability at least $\frac{3}{4}$ a node that is the root of a \emph{small subtree} is chosen, i.e., a subtree that is smaller than $2 \cdot \min\{\sqrt{|T_1|}, \sqrt{|T_2|}\}$.
This however follows readily from Lemma~\ref{lem:eq_split_existence}: since all subtrees at level $h'$ are chosen for exploration with uniform probability, we can conclude that choosing a node that is the root of such a small subtree in $T_1$ has a probability of at least $\frac{3}{4}$. 
From the maximality of $h'$ we can conclude that $s \geq 4 \cdot \min\{\sqrt{|T_1|}, \sqrt{|T_2|}\}$ (see proof of Lemma~\ref{lem:eq_split_existence}).
Hence, Corollary~\ref{lem:eq_split_leaf} ensures that the algorithm terminates with a balanced split when choosing a node that is the root of a small subtree.

Furthermore, note that before level $h'$ is reached, it is not possible for Algorithm~\ref{alg:lasvegas_bidirectional} to return a split that is not a balanced split since the cost of probing is smaller than the bound for balanced splits.
\end{proof}
At level $h'$ Algorithm~\ref{alg:lasvegas_bidirectional} terminates with probability $\frac{3}{4}$.
Careful inspection of the proof of Lemma~\ref{lem:eq_split_existence} and Lemma~\ref{lem:eq_split_compute} reveals that Algorithm~\ref{alg:lasvegas_bidirectional} also terminates with probability at least $\frac{3}{4}$ after every consecutive doubling of $s$.
While the cost (and therefore potential execution time) doubles, the probability of terminating before reaching the respective cost quarters, which defines a geometric series: this results in an expected runtime of Algorithm~\ref{alg:lasvegas_bidirectional} bounded by
\[\mathcal{O}(d \cdot \min\{\sqrt{|T_1|}, \sqrt{|T_2|}\}).\]
We now run Algorithm~\ref{alg:lasvegas_bidirectional} and Algorithm~\ref{alg:deterministic_iso} in series, which results in the desired algorithm: if Algorithm~\ref{alg:lasvegas_bidirectional} terminates with non-isomorphism we are done and otherwise Algorithm~\ref{alg:deterministic_iso} tests isomorphism with the provided split.
We observe that whenever Algorithm~\ref{alg:lasvegas_bidirectional} terminates with a split, the costs of the split are also bounded by $s$: the execution time can not be larger than the cost of the returned split.  
Running the previously described modification of Algorithm~\ref{alg:deterministic_iso} with a split of cost $s$ incurs expected cost bounded by $\mathcal{O}(\log_2(\max\{\sqrt{|T_1|}, \sqrt{|T_2|}\}) \cdot s)$.
Using this, we can conclude the section with the following theorem.
\begin{theorem} Let $T_1, T_2$ be black box search trees with maximum degree $d$. There exists an algorithm for the isomorphism exploration problem with no error that has an expected worst-case runtime bounded by 
\[\mathcal{O}(d \cdot \log_2(\max\{\sqrt{|T_1|}, \sqrt{|T_2|}\}) \cdot \min\{\sqrt{|T_1|}, \sqrt{|T_2|}\}).\]
\end{theorem}

\section{Lower Bounds}
Now we prove lower bounds within the confines of the model. Easy lower bounds can be obtained by considering input trees of height 1, however, we are interested in bounds that also apply to trees of bounded degree. We utilize the search tree family $\mathcal{M}_h$ for this purpose (see Figure~\ref{fig:mk}). A tree is in $\mathcal{M}_h$, if it is a complete binary tree of height $h$ such that leaves have pair-wise distinct colors, i.e., for all $(l_1, l_2) \in L(V(\mathcal{M}_h))^2$ with $l_1 \neq l_2$ it holds that $\Col(l_1) \neq \Col(l_2)$. 

We remark that shrunken multipedes (see \cite{DBLP:conf/esa/NeuenS17}) are graphs that produce search trees very similar to those in $\mathcal{M}_h$ when used as input for individualization-refinement algorithms.

\begin{figure}
	\centering
	\begin{tikzpicture}[
		every node/.style = {minimum width = 0.7em, inner sep = 0, outer sep = 0, draw, circle, fill=gray!50},
		level/.style = {sibling distance = 35mm/#1, level distance=7mm}
		]
		\node[thick, fill=white] {}
		child {	node[thick, fill=white] {}
				child {	node[thick, fill=white] {}
						child {	node[thick, fill=orange] {}
						}
						child {	node[thick, fill=lightblue] {}
						}
				}
				child {	node[thick, fill=white] {}
						child {	node[thick, fill=red!50] {}
						}
						child {	node[thick, fill=yellow!50] {}
						}
				}
		}
		child {	node[thick, fill=white] {}
				child {	node[fill=white, thick] {}
						child {	node[thick, fill=green!50] {}
						}
						child {	node[thick, fill=purple!50] {}
						}
				}
				child {	node[fill=white, thick] {}
						child {	node[thick, fill=lightgray] {}
						}
						child {	node[thick, fill=cyan!50] {}
						}
				}
		};
	\end{tikzpicture}
	\caption{A search tree from the class~$\mathcal{M}_3$.} \label{fig:mk}
\end{figure}
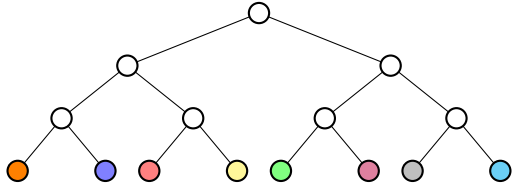

Generally, due to their uniformity,  trees from $\mathcal{M}_{h}$ can only be distinguished or proven isomorphic by considering leaves. A traversal strategy must either conclude -- with good probability ($\frac{1}{2}$)-- that the set of leaves of the trees are entirely disjunct or equal. In the case when trees are isomorphic, the traversal strategy must provide two leaves with equal colors.

\subsection{Randomized Lower Bound}

We prove lower bounds for the isomorphism problem for randomized algorithms that err. 

We will use a particular type of exploration algorithm for our purposes. We call an algorithm \emph{unadaptive} on a class of inputs, if
on each input from the class, the number of queries is always the same (in particular independent of randomness involved in the algorithm) and
the queries performed by the algorithms on inputs from the class are independent of the answers given by the oracle. (The queries may still depend on the randomness involved in the algorithm.)
This means in particular that even when matching leaves have been found the algorithm will simply continue to run, possibly making further queries, and at some later point make a decision about the output. 

\begin{lemma}\label{lem:nonadaptive:suffices}
If some (possibly randomized) algorithm~$A$ solves the isomorphism exploration problem with expected run-time~$f(n)$ and error-probability~$\epsilon$ then for each~$h\in \mathbb{Z}$ there is a randomized algorithm~$B$ that is unadaptive on the class of inputs~$\mathcal{M}_h$ with a run-time in~$\mathcal{O}(f(n))$ and error-probability~$\epsilon$.
\end{lemma}
\begin{proof}

If an algorithm~$A$ solving the problem with expected run-time~$f(n)$ and error probability~$\epsilon$ is given, then by repeating the algorithm and using Markov's inequality we can design an algorithm~$A'$ with a run-time bounded by~$\mathcal{O}(f(n))$ (not just in expectation) that still has an error probability of~$\epsilon$.
For this note that even if the trees have been partially explored, it is possible to simulate the algorithm from scratch by pretending that explored nodes of the tree are unexplored.

To obtain the algorithm~$B$ we alter algorithm~$A'$ by simply pretending all discovered leaves have a randomly chosen previously unused color. More precisely, when a leaf is discovered, we pretend it has a color in~$\{1,\ldots,2^h\}$ drawn independently and uniformly at random from the colors that have not been used yet.  
We continue the simulation until~$A'$ halts. We then claim the input to be a yes instance if we found matching leaves and a no instance otherwise. This can only decrease the error probability in comparison to~$A'$.
\end{proof}

For our lower bound we define a combinatorial problem of trees. Let~$M_h$ be the complete binary tree of height $h$, so that trees in~$\mathcal{M}_h$ are colored versions of~$M_h$. For two rooted trees~$U,S$ let~$\inj(U,S)$ be the set of root respecting injective homomorphisms from~$U$ to~$S$. That is, the set contains the injective maps from~$V(U)$ to~$V(S)$ that map the root of~$U$ to the root of~$S$ and that map an edge of~$U$ to an edge of~$S$.

From now on fix a height~$h$ and consider the tree~$M_h$. Let~$\mathcal{U}_h$ be the set of trees~$U$ for which~$\inj(U, M_h)$ is non-empty. This set contains exactly the trees isomorphic to a subtree of~$M_h$.
For two trees~$U_1,U_2\in \mathcal{U}_h$ we let~$P(h,U_1,U_2)$ be the probability that for uniformly chosen~$\alpha_1\in \inj(U_1,M_h)$ and independently, uniformly chosen~$\alpha_2\in \inj(U_2,M_h)$ the set~$L(M_h)\cap \alpha_1(V(U_1)) \cap \alpha_2(V(U_2))$ is non-empty.
For integers~$a,b$ define \[P(h,a,b)=\max \left\{P(U_1,U_2)\mid |L(U_1)|=a \wedge |L(U_2)|= b\right\}.\] 
Let~$P(h,m)= \max \{P(h,a,b)\mid  a+b\leq m \}$. We will argue that~$P(h,m)$ constitutes an upper bound on the probability of success for a randomized algorithm for isomorphism exploration that queries at most $m$ nodes.

\begin{lemma}
Let~$B$ be an algorithm that is unadaptive on the class of inputs from~$\mathcal{M}_h$. Suppose on inputs from~$\mathcal{M}_h$ algorithm~$B$ makes~$m$ queries and has error probability~$\epsilon$. Then~$1-\epsilon \leq P(h,m)$. \end{lemma}

\begin{proof}
Consider the behavior of algorithm~$B$ on inputs from~$\mathcal{M}_h$ with the colors~$\{1,\ldots,2^h\}$ being randomly assigned bijectively to the leaves.
 The algorithm~$B$ explores subtrees~$T'_1$ and~$T'_2$, one in each of the input trees. Since the algorithm makes~$m$ queries, together these trees can have at most~$m$ leaves.
Our argument groups the possibilities in which~$B$ can query the oracle according to the topology of the two subtrees.

For two trees~$U_1$ and~$U_2$ consider the event~$E_{U_1,U_2}$ that~$T'_1$ is isomorphic to~$U_1$ and~$T'_2$ is isomorphic to~$U_2$. The event can of course only occur if~$|U_1|+|U_2| \leq m$.
Recall that algorithm~$B$, being unadaptive, does not use the information on colors of the leaves provided by the oracle until the very end. Thus, 
the probability that~$B$ finds matching leaves on isomorphic inputs conditional to event~$E_{U_1,U_2}$ is~$P(h,U_1,U_2)$.

We conclude that the probability that~$B$ finds matching leaves\footnote{In our problem definition the algorithm has to find two leaves of the same color. If the task only asked to decide whether the graphs are isomorphic, the algorithm could still guess, which would incur another factor of~$1/2$.} is at most~$P(h,m)$.
\end{proof}

We now show that the trees need to have sufficiently many leaves for~$P(h,T_1,T_2)$ to be large.

\begin{lemma}\label{lem:lower:bound:tree:homs}
$P(h,a,b) \leq \frac{ab}{2^h}$. 
\end{lemma}
\begin{proof}
For two trees~$T_1$ and~$T_2$ let~$E(T_1,T_2)$ be the expected number of elements contained in the set $L(\mathcal{M}_h)\cap \alpha_1(V(T_1)) \cap \alpha_2(V(T_2))$, where~$\alpha_1$ and~$\alpha_2$ are taken independently and uniformly from~$\inj(T_1,\mathcal{M}_h)$ and~$\inj(T_2,\mathcal{M}_h)$, respectively. We define~$E(h,a,b)$ in analogy to~$P(h,a,b)$ as the maximum~$E(T_1,T_2)$ over all choices of~$T_1$ and~$T_2$ with~$|L(T_1)|=a$ and~$|L(T_2)|=b$.
By the Markov inequality it suffices to show that~$E(h,a,b)\leq \frac{ab}{2^h}$.

Only vertices that are of distance~$h$ from the root in~$T_i$ can be mapped to a vertex in~$L(\mathcal{M}_h)$.
The automorphism group of~$\mathcal{M}_h$ can map each leaf to every other leaf (i.e., acts transitively on the leaves). The graph~$\mathcal{M}_h$ has~$2^h$ leaves. 
Thus, for vertices~$v_1\in T_1$ and~$v_2\in T_2$ both of distance~$h$ from the root, the probability that~$\alpha(v_1) = \alpha(v_2)$ is at most~$\frac{1}{2^h}$. 

By linearity of expectation the expected number of pairs~$(v_1,v_2)$ for which~$\alpha(v_1) = \alpha(v_2)\in L(\mathcal{M}_h)$ is at most~$\frac{1}{2^h}\cdot a\cdot b$.
\end{proof}

\begin{theorem}[randomized lower bound]\label{thm:lowerboundiso} In the black box search tree model, a (possibly randomized making errors) traversal strategy runs in $\Omega(\min\{\sqrt{|\eT_1|}, \sqrt{|\eT_2|}\})$ worst-case cost for the isomorphism exploration problem, even on binary trees.
\end{theorem} 
\begin{proof}
By Lemma~\ref{lem:nonadaptive:suffices}, it suffices to show the statement for an unadaptive algorithm~$B$ on $\mathcal{M}_h$.
By Lemma~\ref{lem:lower:bound:tree:homs}, if~$B$ queries less than~$\frac{1}{2}\sqrt{|\mathcal{M}_h|}$ nodes then both trees~$T_1$ and~$T_2$ uncovered by~$B$ have at most~$\frac{1}{2}\sqrt{|\mathcal{M}_h|}$ leaves. But by the previous lemma we know that $P(h,\frac{1}{2}\sqrt{|\mathcal{M}_h|},\frac{1}{2}\sqrt{|\mathcal{M}_h|}) \leq \frac{1}{4}$, which shows that the probability that~$B$ finds matching leaves in the two trees is at most~$\frac{1}{4}$. This shows that~$B$ cannot find matching leaves with probability~$\frac{1}{2}$.
\end{proof}

\subsection{Deterministic Lower Bound}

We exploit the randomized lower bound to obtain a strengthened deterministic one. 

\begin{theorem}[deterministic lower bound] In the black box search tree model, a deterministic traversal strategy runs in $\Omega(\min\{{|\eT_1|}, {|\eT_2|}\})$ worst-case cost for the isomorphism exploration problem, even on binary trees.
\end{theorem} \label{thm:lowerboundiso:determ}

\begin{proof}
Consider a deterministic algorithm on inputs from $\mathcal{M}_{2h}$, where~$h = \log (n)$ and~$n$ is a power of~$2$. By Theorem~\ref{thm:lowerboundiso} there are instances consisting of pairs of trees~$T_1$,~$T_2$ on which the algorithm makes~$\Theta(\sqrt{2^{2\log(n)}})=\Theta(n)$ queries in total. We know from the proof that the trees~$T_i$ can be chosen from~$\mathcal{M}_{2h}$.
For each~$i\in \{1,2\}$, remove from~$T_i$ all non-root vertices whose parents have not been explored (and thus who have not been explored either). Let~$T'_i$ be the resulting tree, respectively for each~$i$.  On the input pair~$(T'_1,T'_2)$ the algorithm behaves exactly the same as on~$(T_1,T_2)$ and thus also makes~$\Theta(n)$ queries in total, however~$T'_i$ has at most~$\mathcal{O}(n)$ vertices. This shows that on~$(T'_1,T'_2)$  the algorithm makes 
$\Omega(\min\{{|T'_1|}, {|T'_2|}\})$ queries.
\end{proof} 

Note that balanced splits for the trees of $\mathcal{M}_h$ can be found almost trivially: after finding out the height $h$ through a single walk, an arbitrary node at level $\frac{h}{2}$ will induce a balanced split.
This shows that while $\mathcal{M}_h$ constitutes worst-case examples for probabilistic algorithms, this is not true for deterministic algorithms. And indeed, our deterministic lower bounds applies subtrees of trees in $\mathcal{M}_h$ which have leaves on different levels.

\section{Motivation Behind the Model}\label{sec:motivation}
The motivation behind the specifics of our model lies in so-called \emph{individualization-refinement} algorithms, the prevailing method to solve the graph isomorphism problem in practice. We will explain that and why the isomorphism exploration problem captures the runtime of these algorithms, but we will not go into detail of how such algorithms work in general, and rather refer to~\cite{McKay:userguide,McKay201494}.

In fact, currently all practical state-of-the-art tools are based on the individualization-refinement paradigm. Invariably, these algorithms perform a type of backtracking procedure to explore the structure of input graphs. As usual with backtracking procedures, this leads to 
a \emph{search tree}. While work is performed by the algorithms in each node of the search tree, the dominating factor for the running time is the size of the tree itself. Specifically, the running time per node is almost linear, coming from the color-refinement algorithm (also known as 1-dimensional Weisfeiler-Leman Algorithm). However, the size of the search tree is exponential in the worst case~\cite{DBLP:conf/esa/NeuenS17}. 

When solving for isomorphisms of two graphs we get two search trees, one for each of the graphs. The leaves of the search trees correspond to complete invariants (described in more detail below), i.e., colors in our terminology. The task of finding isomorphisms in the graphs translates to finding pairs of leaves of equal color in the trees.

As mentioned earlier, nowadays there are various software packages implementing the paradigm in different flavors. These packages differ in many details (see below), of which the most crucial aspect is the traversal strategy through the tree.

\paragraph{The Search Tree.} While the problem definition, algorithms, and lower bounds in this paper do not require further knowledge on how IR-algorithms operate, we want to give some intuition about the composition of the search tree and in particular where the axiom regarding complete isomorphism invariance comes from.

The individualization-refinement search tree is the recursion tree of a backtracking procedure whose goal it is to analyze the structure of an input graph (two input graphs in case of the isomorphism problem).
Initially, forming the root of the tree, the algorithm distinguishes the vertices of the input graph using readily computed invariants. For example the vertices are easily distinguished by their degree, but other information is also used. The specific algorithm typically used for this is the color \emph{refinement algorithm} (also known as vertex classification or 1-dimensional Weisfeiler-Leman algorithm). 
In case all vertices are distinguished from another, isomorphism can easily be checked, and no recursion is needed.
Otherwise the algorithm starts to pick a class of indistinguishable vertices and for each vertex in this class, one at a time, 
artificially alters the vertex to make it distinguishable. This process is called \emph{individualization} (hence the name individualization-refinement algorithm). Each individualization causes a recursive call, which corresponds to a child of the current node in the search tree.
In the recursion, the refinement algorithm is called again to check if new information propagates through the graph and can be used to distinguish vertices further. As the entire procedure proceeds recursively it allows us to distinguish more and more vertices from each other. 
Recursion continues until all vertices have been distinguished from one another, i.e., the partition of the vertices into different types is \emph{discrete}. The leaves of the individualization-refinement search tree therefore correspond to discrete partitions.
To the leaves, because all vertices have been distinguished, we can associate an invariant that completely describes the structure. We call this a complete invariant. This invariant corresponds to the color of the leaf in our exploration model.

The defining property of these algorithms is that all computations are made in an isomorphism-invariant fashion. This is precisely what leads to our invariance axiom. Indeed, 
checking isomorphism between two leaves $l_1, l_2$ in the backtracking tree then becomes trivial since in each leaf all vertices have been distinguished from one another. In particular there can be only one possible isomorphism, which would have to map vertices of the same type to each other.
If an isomorphism exists this implies that the individualization choices made to get to $l_1$ can be mapped to the choices made to get to $l_2$ with the isomorphism.

The isomorphism-invariance of the entire procedure guarantees that for matching leaves there is an automorphism (isomorphism when considering trees from different input graphs) mapping one leaf to the other. In fact, the number of occurrences of each leaf color is exactly the order of the automorphism group of the input graph.

While, depending on the input, the backtracking trees can come in many shapes and sizes, we want to record several properties. For starters internal vertices cannot have only one child since there is no reason to individualize vertices in singleton classes. Furthermore the number of children a node in the search tree can have is certainly bounded by the number of vertices of the input graph. However, in practice most nodes have significantly fewer children.

The size of the search tree is the dominating factor of the running time. In fact, for most implementations the non-recursive work can 
be polynomially bounded. Thus, up to a polynomial factor, the running time agrees with the number of vertices of the search tree that are traversed (see for example \cite[Theorem 9]{thesis} and \cite{DBLP:journals/corr/abs-0804-4881}). In the worst case the size of the search tree and thus the running time of IR-algorithms is exponential in the size of the input~\cite{DBLP:conf/stoc/NeuenS18}.

\paragraph{Practical Heuristics.} A closer look at the practical tools reveals that these algorithms often do not have to traverse the entire tree. In particular there are two heuristics commonly applied to prune parts of the search tree. 
They are called \emph{invariant pruning} and \emph{automorphism pruning}. 
The reader familiar with isomorphism-testing tools may worry that our model omits the two crucial aspects of IR-algorithms. However, 
in the following, we explain why both of these pruning techniques are already captured by the exploration model we have given.

First of all, by using \emph{invariant pruning}, individualization-refinement algorithms are sometimes able to cut off parts of the search tree at inner nodes. This is done using invariants which yield different values at different parts of the tree. 
In principle this process could be emulated
in the black box search trees by adding colors to the inner nodes. 
When doing so, it is crucial that the invariance axiom does not apply to the inner nodes, since only for leaves the associated invariants are complete. We could extend all results from this paper to the adapted model.
In any case, it turns out that to some extent this kind of mechanism is already  captured in our model since it allows inner vertices to have different degrees. Degrees of vertices can serve to distinguish inner nodes exactly in the same fashion as invariants do.

Second of all, by using discovered automorphisms of the graph, \emph{automorphism pruning} is a method to skip branches that we already know are symmetric. Another way of seeing this is that automorphisms allow us to form quotients of the search tree. In our setting, we can simply define the trees to be the quotients of the original search trees. This simulates perfect automorphism pruning. It does not explain how to find one or all automorphisms in one of the trees, but these kinds of problems are closely linked to finding isomorphisms and can also be expressed in our search tree model.
Automorphism problems and their relation to our model are discussed further below.

In summary, both types of heuristics are captured by our model.

\paragraph{Traversal Strategies in Practical Tools.}
There are only two main traversal strategies used by competitive practical tools.
In fact, with the exception discussed below, all competitive tools essentially traverse the search tree using depth-first search \cite{Darga:2004:ESS:996566.996712,JunttilaKaski:ALENEX2007,DBLP:journals/jam/Lopez-PresaCA14,McKay201494}.  
As an all-purpose traversal strategy it allows for the discovery of automorphisms used for automorphism pruning as described above and invariant pruning can also be applied effectively.

However, the practical solver \Traces{} introduced a radically different strategy, which turns out to be much more effective in most practical cases:
breadth-first traversal is combined with random walks through the search trees \cite{McKay201494,DBLP:journals/corr/abs-0804-4881}.
The idea is that breadth-first traversal is used to maximize the applicability of pruning rules, while random walks are used to discover the automorphisms for automorphism pruning.
While crucially exploiting randomization, the tool does not make errors. In practice memory consumption is an issue for breath-first search, but newer versions of \Traces{} have the means for a better handling of memory.

\paragraph{Related Problems.}
Practical algorithms mostly do not decide the graph isomorphism problem, but rather solve one of two strongly related types of problems. 
In the following, we discuss these related problems and how they can be modeled using black box search trees.

First of all, the \emph{automorphism group problem} requires computation of the entire automorphism group of a given graph. This problem is closely related to the graph isomorphism problem and there are polynomial-time Turing reductions from each problem to the other~\cite{DBLP:journals/ipl/Mathon79}.
Regarding individualization-refinement search trees, in our model the problem corresponds to finding all leaves of some color chosen by the algorithm.

For randomized algorithms, another problem of interest is to find a \emph{uniform automorphism}, i.e., to compute an element of the automorphism group, but the outcome has to be uniformly distributed among all elements of the group. Repeatedly and independently finding uniformly random automorphisms allows us to generate the entire group and solve the automorphism group problem.
In our model this translates to finding two leaves of the same color in one tree so that the second leaf is independent of the first.

Alternatively, we can solve the \emph{asymmetry problem} or equivalently the problem of searching for a \emph{non-trivial graph automorphism}. In our model this translates to finding two distinct leaves of the same color if they exist.
While the asymmetry problem reduces to the isomorphism problem, a reduction in the other direction is not known. 
However, when it comes to IR-algorithms, it does suffice to solve the asymmetry problem. When non-trivial automorphisms are discovered, a wrapper algorithm may then remove the symmetry from the search tree and repeat the task on a quotient of the search tree. Such a wrapper algorithm will have to repeat the problem at most a number of times that is quasi-linear in the order of the graph. This is because the repetitions are directly tied to the length of the longest subgroup series of the automorphism group, which can only be quasi-linear in the order of the graph.

Overall, our results on traversal strategies immediately carry over to the various tasks regarding automorphism group computations.  

Another important problem in practice is \emph{canonization}. The goal here is essentially to find a normal form for graphs by finding a canonical ordering of the vertices.
This way isomorphism testing reduces to equality testing of the normal forms.
Hence, graph isomorphism reduces to canonization in polynomial-time, but we currently do not know whether graph isomorphism and canonization are polynomial-time equivalent.

In our context of IR-algorithms and their backtracking trees, this problem can be modeled as follows. The solvers are now executed on a single graph, yielding a single search tree. The goal is to return a particular leaf in the given input trees. The requirement is that the output has to be consistent across different inputs. That means that for different but isomorphic inputs, leaves of the same color have to be returned.

It is not clear to us whether any of the techniques for sublinear exploration developed in this paper can be transferred to the canonization problem.

\section{Conclusion and Future Work}
We designed an abstract model that resembles the backtracking behavior of IR-algorithms and proved bounds for various scenarios.
We want to stress the fact that the class of trees $\mathcal{M}_h$ used throughout the paper for lower bound constructions models actual recursion trees of the IR-algorithms. In fact the trees closely resemble those arising form so-called shrunken multipede graphs of \cite{DBLP:conf/esa/NeuenS17}, which form worst case inputs for all IR-algorithms. These recursion trees are in particular of degree at most 4 (actually 2 for various IR-algorithms) and have exponential size. In other words, our worst case lower bounds apply to instances stemming from true inputs to IR-algorithms.

Using our new insights we can explain why some of the strategies used by the currently fastest practical solver \Traces{} turn out to be highly efficient.
As discussed previously, \Traces{} uses breadth-first search intertwined with random walks of the search tree. 
In particular, this is often done in a cost balancing manner, such that the number of random walks is proportional to the cost of breadth-first search.
This in turn often leads to the automorphism group being found in time proportional to the square root of the search tree size.
For sophisticated pieces of software such as \Traces{}, the traversal strategy is of course not the only deciding factor when it comes to running time. 
However, generally, the experimental paths often enable \Traces{} to discover automorphisms much earlier than solvers solely utilizing deterministic depth-first traversal. Hence automorphisms are available more quickly for pruning.
Overall, in some sense, \Traces{} emulates some of the techniques described in our Monte Carlo algorithm. 

Interestingly, \Traces{} also sometimes uses some techniques of the Las Vegas algorithm we describe. Specifically 
 it performs splits in its ``special traversal'' strategy for automorphism groups (see \cite{McKay201494}). When \Traces{} detects a leaf on level $h$ with parent~$v$, in our terminology it executes the split $(v, h-1)$. Since many graphs in the benchmark suite of \cite{McKay201494} have search trees of height~$2$ or~$3$, \Traces{} in practice turns out to frequently perform splits that are fairly balanced. This results in significant speedups over other solvers (e.g., see runtime on combinatorial graphs with switched edges in \cite{McKay201494}). 

In subsequent work, we were able to show that an implementation of the Monte Carlo approach can indeed outperform state-of-the-art solutions for isomorphism testing in practice. Interestingly, beyond superior worst-case guarantees, the approach has further practical advantages that simplify its implementation over state-of-the-art tools \cite{alenexpaper}. 

Regarding future work, a theoretical question that remains is whether sublinear traversal strategies for the graph canonization problem are possible. Furthermore, the challenge remains to close the gap of logarithmic factors between our upper and lower bounds. Also one might want to address the fact that the size of the larger tree rather than the size of the smaller tree appears in the upper bound of the Las Vegas algorithm.

\bibliography{main}
\bibliographystyle{plain}

\end{document}